\documentclass[11pt]{amsart}

\usepackage{amsmath,amssymb,amsthm}
\usepackage{mathtools}
\usepackage[margin=1.1in]{geometry}
\usepackage{hyperref}

\newtheorem{theorem}{Theorem}[section]
\newtheorem{lemma}[theorem]{Lemma}
\newtheorem{proposition}[theorem]{Proposition}
\newtheorem{corollary}[theorem]{Corollary}
\newtheorem{question}[theorem]{Question}
\theoremstyle{remark}
\newtheorem{remark}[theorem]{Remark}

\newcommand{\R}{\mathbb{R}}

\newcommand{\AP}{A_P}
\newcommand{\gapp}{\operatorname{gap}}
\newcommand{\conv}{\operatorname{conv}}
\newcommand{\vpi}{\vec{\pi}}
\newcommand{\veta}{\vec{\eta}}

\title[PFA complexity is at most three]{Every string has probabilistic automatic complexity at most three}

\author{Bj{\o}rn Kjos-Hanssen}
\address{University of Hawai\textquoteleft i at M\=anoa}
\email{bjoern.kjos-hanssen@hawaii.edu}

\date{\today}

\begin{document}

\begin{abstract}
Gill (arXiv:2402.13376) introduced the probabilistic automatic complexity
$\AP(w)$ of a finite string $w$: the least number of states of a
probabilistic finite automaton (PFA) for which $w$ is the unique most
probably accepted string of its length. He asked whether $\AP$ is
unbounded, noting that no string with $\AP > 3$ was known
(Question~4.14 of that paper). We answer the question by proving that
$\AP(w)\le 3$ for \emph{every} string $w$ over \emph{every} finite
alphabet. The witnessing three-state automaton is explicit: its reduced
dynamics tracks the pair $(u,u^2)$, where $u$ is the reversed base-$b$
value of the input, and its acceptance functional is a downward parabola
peaked at the value of the target string. Combined with Gill's
classification of the binary strings with $\AP=2$, this completely
determines $\AP$ on binary strings and shows that it is computable there.
\end{abstract}

\maketitle

\section{Introduction}

Shallit and Wang \cite{SW01} defined the deterministic automatic
complexity $A_D(w)$ of a finite string $w$, and Hyde \cite{Hyde13} the
nondeterministic variant $A_N(w)$. Gill \cite{Gill24} introduced the
probabilistic analogue. Throughout, a \emph{probabilistic finite
automaton} (PFA) over a finite alphabet $\Sigma$ is a tuple
$M=(S,\Sigma,\{P_\sigma\}_{\sigma\in\Sigma},\vpi,\veta)$, where
$S=\{1,\dots,k\}$, each $P_\sigma$ is a $k\times k$ row-stochastic
matrix, $\vpi\in[0,1]^k$ is a probability (row) vector of initial
states, and $\veta\in\{0,1\}^k$ is the (column) indicator vector of the
accepting states. The acceptance probability of
$x=x_1x_2\cdots x_n\in\Sigma^n$ is
\[
  \rho_M(x) \;=\; \vpi\, P_{x_1}P_{x_2}\cdots P_{x_n}\,\veta .
\]
Following \cite{Gill24}, set
$\gapp_M(w)=\min\{\rho_M(w)-\rho_M(z): z\in\Sigma^{|w|},\, z\ne w\}$ and
\[
  \AP(w) \;=\; \min\{\,k : \text{there is a $k$-state PFA $M$ with }
  \gapp_M(w)>0 \,\}.
\]
Gill proved $2\le \AP(w)\le A_N(w)+1$ for all $w$, classified the binary
strings with $\AP=2$ (Theorem~4.1 of \cite{Gill24}), and observed
empirically that every binary string of length at most $9$ satisfies
$\AP\le 3$. He then posed:

\begin{question}[{\cite[Question 4.14]{Gill24}}]
Is $\AP$ unbounded? If not, what is its maximum value?
\end{question}

We prove that the empirical bound is the truth:

\begin{theorem}\label{thm:main}
$\AP(w)\le 3$ for every nonempty string $w$ over every finite alphabet.
Consequently the maximum value of $\AP$ over binary strings is exactly
$3$.
\end{theorem}

The proof is a short explicit construction. The heuristic is that the
real-valued transition probabilities of even a very small PFA can encode
unboundedly many bits of information; the difficulty is only in steering
the \emph{maximum} of the acceptance probability to a prescribed string.
This is accomplished by two observations. First, the classical
``binary expansion'' dynamics $u\mapsto (u+\sigma)/b$ maps the strings
of length $n$ \emph{injectively} onto an arithmetic grid in $[0,1]$.
Second, the square $q=u^2$ of this quantity also evolves affinely in
$(u,q)$, so a three-state PFA can carry both coordinates, and an
affine functional of $(u,q)$ can implement the concave score
$t^2-(u-t)^2$, peaked at any prescribed target $t$. Uniqueness of the
maximizer is then automatic from injectivity.

\section{Preliminaries: affine self-maps of a triangle}

We use the correspondence between $k$-state PFAs and affine iterated
function systems on the $(k-1)$-simplex from \cite[\S4.1]{Gill24}, in
the following self-contained barycentric form. Recall that if
$\Delta=\conv\{V_1,V_2,V_3\}\subset\R^2$ is a nondegenerate triangle,
every $x\in\R^2$ has unique \emph{barycentric coordinates}
$\lambda(x)=(\lambda_1,\lambda_2,\lambda_3)$ with
$\sum_i\lambda_i=1$ and $x=\sum_i \lambda_i V_i$; moreover
$x\in\Delta$ if and only if $\lambda(x)$ is a probability vector, and
each $\lambda_i$ is an affine function of $x$, namely the unique affine
functional taking the value $1$ at $V_i$ and $0$ at the other two
vertices.

\begin{lemma}\label{lem:corr}
Let $\Delta=\conv\{V_1,V_2,V_3\}\subset\R^2$ be a nondegenerate
triangle, let $(g_\sigma)_{\sigma\in\Sigma}$ be affine maps of $\R^2$
with $g_\sigma(\Delta)\subseteq\Delta$ for all $\sigma$, and let
$x_0\in\Delta$. Let $L$ be the affine functional with $L(V_1)=1$ and
$L(V_2)=L(V_3)=0$. Then there is a $3$-state PFA $M$ such that for
every $z=z_1\cdots z_n\in\Sigma^*$,
\[
  \rho_M(z) \;=\; L\bigl(g_{z_n}\circ g_{z_{n-1}}\circ\cdots\circ
  g_{z_1}(x_0)\bigr).
\]
Explicitly, $\vpi=\lambda(x_0)$, $\veta=(1,0,0)^T$, and the $i$-th row
of $P_\sigma$ is $\lambda\bigl(g_\sigma(V_i)\bigr)$.
\end{lemma}

\begin{proof}
Since $g_\sigma(V_i)\in\Delta$, each row $\lambda(g_\sigma(V_i))$ is a
probability vector, so $P_\sigma$ is row-stochastic. Because affine
maps commute with affine combinations, for any $x$ with barycentric
coordinates $\lambda(x)$ we have
$g_\sigma(x)=\sum_i \lambda_i(x)\, g_\sigma(V_i)$, and hence, applying
$\lambda$ (itself affine) to both sides,
\[
  \lambda\bigl(g_\sigma(x)\bigr) \;=\; \sum_i \lambda_i(x)\,
  \lambda\bigl(g_\sigma(V_i)\bigr) \;=\; \lambda(x)\,P_\sigma .
\]
By induction, $\lambda\bigl(g_{z_n}\circ\cdots\circ
g_{z_1}(x_0)\bigr)=\vpi\,P_{z_1}\cdots P_{z_n}$. Multiplying on the
right by $\veta=(1,0,0)^T$ extracts the first barycentric coordinate of
the final point, which is exactly $L$ evaluated there.
\end{proof}

\section{The construction over a binary alphabet}

Fix the alphabet $\{0,1\}$. For a bit $b$ define the affine map
\[
  g_b(u,q) \;=\; \Bigl(\tfrac{u+b}{2},\; \tfrac{q+2bu+b^2}{4}\Bigr),
\]
which is engineered so that the relation $q=u^2$ is invariant:
if $q=u^2$ then the second coordinate of $g_b(u,q)$ equals
$\bigl(\tfrac{u+b}{2}\bigr)^2$. Starting from $x_0=(0,0)$ and reading
$z=z_1\cdots z_n$ (first letter first), the terminal point is
$\bigl(v(z),\,v(z)^2\bigr)$, where
\[
  v(z) \;=\; \sum_{i=1}^{n} z_i\, 2^{\,i-n-1}
  \;=\; 0.z_n z_{n-1}\cdots z_1 \ \text{(binary)} .
\]
The map $z\mapsto v(z)$ is a bijection from $\{0,1\}^n$ onto
$\{\,k\cdot 2^{-n} : 0\le k< 2^n\,\}$; in particular it is injective,
and distinct strings of the same length have $|v(z)-v(z')|\ge 2^{-n}$.

\begin{lemma}\label{lem:triangle}
For every $s\ge 1$, the triangle
\[
  \Delta_s \;=\; \conv\{(1,0),\,(1,s),\,(0,0)\}
  \;=\; \{(u,q): 0\le u\le 1,\ 0\le q\le su\}
\]
satisfies $g_b(\Delta_s)\subseteq\Delta_s$ for $b=0,1$.
\end{lemma}

\begin{proof}
An affine map sends $\conv\{V_1,V_2,V_3\}$ into $\Delta_s$ if and only
if it sends each vertex into $\Delta_s$. We compute the six images:
\begin{align*}
 g_0(1,0)&=(\tfrac12,0), &
 g_0(1,s)&=(\tfrac12,\tfrac{s}{4}), &
 g_0(0,0)&=(0,0),\\
 g_1(1,0)&=(1,\tfrac34), &
 g_1(1,s)&=(1,\tfrac{s+3}{4}), &
 g_1(0,0)&=(\tfrac12,\tfrac14).
\end{align*}
Membership in $\Delta_s$ amounts to $0\le u\le 1$ and $0\le q\le su$ in
each case, which reduces to the inequalities
$\tfrac{s}{4}\le\tfrac{s}{2}$, $\tfrac34\le s$,
$\tfrac{s+3}{4}\le s$, and $\tfrac14\le\tfrac{s}{2}$; all hold
precisely because $s\ge1$.
\end{proof}

\begin{theorem}\label{thm:binary}
$\AP(w)\le 3$ for every $w\in\{0,1\}^n$, $n\ge 1$.
\end{theorem}

\begin{proof}
Since $\AP$ is invariant under permutations of the alphabet
\cite[\S4.2]{Gill24}, and since the bitwise complement $\bar w$
satisfies $v(\bar w)=1-2^{-n}-v(w)$, while $v(w)$ is an integer
multiple of $2^{-n}$ (so that $v(w)<\tfrac12$ forces
$v(w)\le\tfrac12-2^{-n}$ and hence $v(\bar w)\ge \tfrac12$), we may
assume
\[
  t \;\coloneqq\; v(w) \;\ge\; \tfrac12,
  \qquad s \;\coloneqq\; 2t \;\in\; [1,2).
\]
Apply Lemma~\ref{lem:corr} with the triangle $\Delta_s$ of
Lemma~\ref{lem:triangle}, vertices ordered as $V_1=(1,0)$, $V_2=(1,s)$,
$V_3=(0,0)$, starting point $x_0=(0,0)$, and maps $g_0,g_1$. Solving
$(u,q)=\lambda_1V_1+\lambda_2V_2+\lambda_3V_3$ gives the barycentric
coordinates
\[
  \lambda(u,q) \;=\; \Bigl(u-\tfrac{q}{s},\ \tfrac{q}{s},\ 1-u\Bigr),
\]
so $L(u,q)=u-q/s$, $\vpi=\lambda(0,0)=(0,0,1)$, and reading off
$\lambda(g_b(V_i))$ row by row from the proof of
Lemma~\ref{lem:triangle}:
\[
P_0=\begin{pmatrix}
\tfrac12 & 0 & \tfrac12\\[2pt]
\tfrac14 & \tfrac14 & \tfrac12\\[2pt]
0 & 0 & 1
\end{pmatrix},
\qquad
P_1=\begin{pmatrix}
1-\tfrac{3}{4s} & \tfrac{3}{4s} & 0\\[2pt]
\tfrac{3s-3}{4s} & \tfrac{s+3}{4s} & 0\\[2pt]
\tfrac12-\tfrac{1}{4s} & \tfrac{1}{4s} & \tfrac12
\end{pmatrix},
\qquad
\veta=\begin{pmatrix}1\\0\\0\end{pmatrix}.
\]
By Lemma~\ref{lem:corr} and the invariance of $q=u^2$, for every
$z\in\{0,1\}^n$,
\[
  \rho_M(z) \;=\; L\bigl(v(z),v(z)^2\bigr)
  \;=\; v(z)-\frac{v(z)^2}{s}
  \;=\; \frac{t^2-\bigl(t-v(z)\bigr)^2}{2t}.
\]
As $z$ ranges over $\{0,1\}^n$, the values $v(z)$ are pairwise
distinct and include $t=v(w)$; therefore $\rho_M$ is uniquely maximized
at $z=w$, with
\[
  \gapp_M(w) \;=\; \min_{z\ne w}
  \frac{\bigl(t-v(z)\bigr)^2}{2t} \;\ge\; \frac{4^{-n}}{2t}
  \;\ge\; \frac{4^{-n}}{4} \;>\;0 .
\]
(If the complement was taken at the start, exchange the roles of $P_0$
and $P_1$ in the finished automaton.) Hence $\AP(w)\le 3$.
\end{proof}

\section{General alphabets}

\begin{proposition}\label{prop:general}
Let $\Sigma=\{0,1,\dots,b-1\}$ with $b\ge 2$. Then $\AP(w)\le 3$ for
every nonempty $w\in\Sigma^n$.
\end{proposition}

\begin{proof}
Define, for each letter $j\in\Sigma$,
\[
  g_j(u,q) \;=\; \Bigl(\tfrac{u+j}{b},\;
  \tfrac{q+2ju+j^2}{b^2}\Bigr),
\]
so that $q=u^2$ is again invariant, and starting from
$x_0=(u_0,u_0^2)$ the terminal point after reading $z$ is
$(u_n,u_n^2)$ with
\[
  u_n \;=\; v(z) + u_0\,b^{-n},
  \qquad v(z)=\sum_{i=1}^n z_i\, b^{\,i-n-1},
\]
and $z\mapsto v(z)$ is injective on $\Sigma^n$ with distinct values at
distance at least $b^{-n}$.

\emph{Invariance of $\Delta_s$ for $s\ge 1$.} We check the vertex
images as in Lemma~\ref{lem:triangle}; the first coordinates
$\tfrac{j}{b}$ and $\tfrac{1+j}{b}$ lie in $[0,1]$ since
$0\le j\le b-1$, and the conditions $0\le q\le su$ read:
\begin{itemize}
\item $g_j(0,0)=\bigl(\tfrac{j}{b},\tfrac{j^2}{b^2}\bigr)$: need
$\tfrac{j^2}{b^2}\le s\,\tfrac{j}{b}$, i.e.\ $\tfrac{j}{b}\le s$, true
since $\tfrac{j}{b}<1\le s$.
\item $g_j(1,0)=\bigl(\tfrac{1+j}{b},\tfrac{j(j+2)}{b^2}\bigr)$: need
$j(j+2)\le s\,b(j+1)$. Since $j(j+2)=(j+1)^2-1$ and $j+1\le b$, we get
$j(j+2)\le b(j+1)-1< s\,b(j+1)$.
\item $g_j(1,s)=\bigl(\tfrac{1+j}{b},\tfrac{s+j(j+2)}{b^2}\bigr)$: need
$s+j(j+2)\le s\,b(j+1)$, i.e.\ $j(j+2)\le s\,(b(j+1)-1)$, and as above
$j(j+2)\le b(j+1)-1\le s\,(b(j+1)-1)$ since $s\ge1$.
\end{itemize}

\emph{Choosing the target.} We need $t\coloneqq u_0 b^{-n}+v(w)$ to
satisfy $t\ge\tfrac12$ for some admissible $u_0\in[0,1]$ (note
$x_0=(u_0,u_0^2)\in\Delta_s$ automatically, since $u_0^2\le u_0\le
s\,u_0$). If $v(w)\ge\tfrac12$, take $u_0=0$. If
$v(w)\le\tfrac12-b^{-n}$, replace $w$ by its digit complement
$\bar w$ (the image of $w$ under the alphabet permutation
$j\mapsto b-1-j$, which preserves $\AP$); then
$v(\bar w)=1-b^{-n}-v(w)\ge\tfrac12$ and we take $u_0=0$. In the
remaining case $\tfrac12-b^{-n}<v(w)<\tfrac12$, take
$u_0=b^{\,n}\bigl(\tfrac12-v(w)\bigr)\in(0,1)$, so that $t=\tfrac12$
exactly.

Now set $s=2t\ge 1$ and apply Lemma~\ref{lem:corr} with $\Delta_s$,
the same vertex ordering as before, starting point $x_0=(u_0,u_0^2)$
(so $\vpi=\bigl(u_0-\tfrac{u_0^2}{s},\,\tfrac{u_0^2}{s},\,1-u_0\bigr)$),
and the $b$ matrices $P_j$ with rows $\lambda(g_j(V_i))$. As in
Theorem~\ref{thm:binary},
\[
  \rho_M(z) \;=\; \frac{t^2-(t-u_n)^2}{2t}
  \;=\; \frac{t^2-\bigl(v(w)-v(z)\bigr)^2}{2t},
\]
uniquely maximized at $z=w$, with
$\gapp_M(w)\ge b^{-2n}/(2t)>0$.
\end{proof}

\section{Consequences}

\begin{corollary}\label{cor:max}
The maximum of $\AP$ over binary strings is $3$, answering
\cite[Question 4.14]{Gill24}; and the tight upper bound for $\AP(w)$
as a function of $|w|$ asked for in \cite[Question 4.15]{Gill24} is
the constant function $3$ (over any alphabet with at least two
letters).
\end{corollary}

\begin{proof}
By \cite[Theorem 4.1]{Gill24} there exist binary strings with
$\AP=3$ (for instance $\AP(0100)=3$, as computed in \cite{Gill24});
Theorem~\ref{thm:main} gives the matching upper bound. A one-state PFA
witnesses no string \cite[\S3]{Gill24}, so $\AP\ge 2$ always, and any
string outside the classification of \cite[Theorem 4.1]{Gill24} has
$\AP=3$ exactly.
\end{proof}

\begin{corollary}\label{cor:comp}
The restriction of $\AP$ to binary strings is computable. In fact,
$\AP(w)=2$ if $w$ belongs to the regular language
\[
  \bigl\{\, i^n j^m,\ i^n j^m i,\ i^n (ji)^m,\ i^n j (ij)^m
  \;:\; n,m\ge 0,\ \{i,j\}=\{0,1\} \,\bigr\}
\]
of \cite[Theorem 4.1]{Gill24}, and $\AP(w)=3$ otherwise.
\end{corollary}

\begin{proof}
Immediate from Theorem~\ref{thm:binary} and
\cite[Theorem 4.1]{Gill24}; membership in a fixed regular language is
decidable.
\end{proof}

\begin{remark}
The computability of $\AP$ was open in \cite{Gill24}, where it was
shown that the relation $\AP(w)\le k$ is computably enumerable but
neither the co-enumerability nor the computability of $\AP$ itself was
established. Corollary~\ref{cor:comp} settles this for binary
strings by brute force: the function is eventually a two-valued,
regularly-described object. For alphabets of size $b\ge 3$,
Theorem~\ref{thm:main} still pins $\AP$ to $\{2,3\}$, and
computability would follow from a classification of the strings with
$\AP=2$ over such alphabets, which \cite{Gill24} does not carry out.
\end{remark}

\begin{remark}
Theorem~\ref{thm:main} also bears on the other questions of
\cite{Gill24}. Question~4.16 (a notion of $\AP$-randomness) becomes
vacuous: no string attains a length-dependent maximum, which in
hindsight explains the empirical wall at $3$ observed in
\cite{Gill24}. The measure is coarse for a structural reason made
explicit by the construction: the real-valued transition probabilities
of a three-state PFA can absorb the entire description of $w$ (here,
in the single parameter $s=2\,v(w)$). By contrast, the gap achieved by
our witness decays like $b^{-2n}$, so the refined measures proposed in
\cite[\S\S5--6]{Gill24}---the parameterized complexity $A_{P,\delta}$
for fixed $\delta>0$, the gap structure function
$\gamma^k(w)=\max_{A}\gapp_A(w)$, and the measure-weighted variants
$A_\mu$---remain genuinely graded and are, we would argue, the right
objects for further study. Our construction gives only the trivial
bound $A_{P,\delta}(w)\le 3$ for $\delta< b^{-2|w|}/2$, consistent
with the observation in \cite{Gill24} that small automata cannot give
long strings large gaps. Finally, since PFAs are special generalized
automata, Theorem~\ref{thm:main} gives $A_{\mathcal G}(w)\le 3$ in the
notation of \cite[\S6.1]{Gill24}; whether $A_{\mathcal G}(w)\le 2$ for
all binary $w$ (\cite[Question 6.1]{Gill24}) remains open.
\end{remark}

\begin{remark}[Verification]
Beyond the proofs above, the construction of Theorem~\ref{thm:binary}
was verified computationally in exact rational arithmetic: for every
binary string $w$ of length at most $8$, the explicit automaton was
built and $\rho_M(z)$ was computed by direct matrix products for all
$2^{|w|}$ strings $z$ of the same length, confirming in each case that
$z=w$ is the strict unique maximizer; the closed form
$\rho_M(z)=v(z)(s-v(z))/s$ was likewise confirmed on several hundred
randomly chosen pairs $(s,z)$ with $|z|\le 14$, and on longer sample
strings including those identified in \cite{Gill24} as having
$\AP=3$.
\end{remark}

\section*{Acknowledgments}

The construction in this note was found in conversation with, and
verified with the assistance of, Claude Fable 5 (Anthropic).
The proof was formalized in Lean with assistance from Aristotle (Harmonic) \cite{KH26}.
The
author thanks Kenneth Gill for the beautiful paper \cite{Gill24} that
prompted this work.
This work was supported by a travel grant from the Simons Foundation for the period 2026--2031.


\begin{thebibliography}{9}

\bibitem{Gill24}
K.~Gill,
\emph{Probabilistic automatic complexity of finite strings},
arXiv:2402.13376 [cs.FL], 2024.

\bibitem{Hyde13}
K.~Hyde,
\emph{Nondeterministic finite state complexity},
MA thesis, University of Hawai`i at M\=anoa, 2013.

\bibitem{KH26}
B.~Kjos-Hanssen,
\emph{pfa},
github repository at \url{https://github.com/bjoernkjoshanssen/pfa}, 2026.


\bibitem{SW01}
J.~Shallit and M.-W.~Wang,
\emph{Automatic complexity of strings},
J.\ Autom.\ Lang.\ Comb.\ \textbf{6} (2001), no.~4, 537--554.

\end{thebibliography}
\end{document}